\documentclass[prx,twocolumn,letterpaper, superscriptaddress, eqsecnum, nofootinbib]{revtex4-2}

\usepackage[colorlinks=true,linkcolor=blue, citecolor=magenta,urlcolor=magenta,breaklinks=true, pdftex]{hyperref}
\usepackage{amsmath}
\usepackage{amsfonts}
\usepackage{graphicx}
\usepackage{multirow}
\usepackage[linesnumbered, ruled, lined, commentsnumbered]{algorithm2e} 
\usepackage{amsthm}
\newtheorem{prop}{Proposition}[section]

\newcommand{\ii}{i}

\begin{document}

\title{Protection of Exponential Operation using Stabilizer Codes in the Early Fault Tolerance Era}

\author{Dawei Zhong}
\affiliation{Department of Physics \& Astronomy, University of Southern California, Los Angeles, California 90089, United States}
\author{Todd A. Brun}
\affiliation{Department of Physics \& Astronomy, University of Southern California, Los Angeles, California 90089, United States}
\affiliation{Ming Hsieh Department of Electrical \& Computer Engineering, University of Southern California, Los Angeles, California 90089, United States}
\affiliation{Department of Computer Science, University of Southern California, Los Angeles, California 90089, United States}

\begin{abstract}
    Quantum error correction offers a promising path to suppress errors in quantum processors, but the resources required to protect logical operations from noise, especially non-Clifford operations, pose a substantial challenge to achieve practical quantum advantage in the early fault-tolerant quantum computing (EFTQC) era. In this work, we develop a systematic scheme to encode exponential maps of the form $\exp(-\ii\theta P)$ into stabilizer codes with simple circuit structures and low qubit overhead. We provide encoded circuits with small first-order logical error rate after postselection for the $[[n, n-2, 2]]$ quantum error-detecting codes and the $[[5, 1, 3]]$, $[[7, 1, 3]]$, and $[[15, 7, 3]]$ quantum error-correcting codes. Detailed analysis shows that under the level of physical noise of current devices, our encoding scheme is 4--7 times less noisy than the unencoded operation, while at most $3\%$ of runs need to be discarded.
\end{abstract}

\maketitle

\section{Introduction}\label{sec:intro}
Quantum computing promises significant speed-ups for solving classically intractable problems~\cite{nisq} such as integer factorization~\cite{factoring} and Hamiltonian simulation~\cite{qsim}. However, due to the presence of noise, it is challenging to demonstrate quantum advantage in solving these problems on current quantum computers. Quantum error-correction codes (QECCs) and fault-tolerant quantum computing (FTQC)~\cite{shor_code, steane_code, css_code, preskill_FTQC} are promising strategies to address the issue of noise. However, even though the capabilities of QECCs have been demonstrated in several recent experiments~\cite{willow,quera,tesseract, atomcomputing2024}, current and near-term quantum computers still cannot protect useful quantum algorithms from noise with FTQC due to the limitations of current hardware. 

To make the best use of hardware and realize practical quantum advantage, one strategy is to optimize quantum error correction codes to achieve near (or even full) fault tolerance with as few computational resources as possible. In particular, finding an approach to suppress noise in non-Clifford operations with low overhead is of the greatest importance for achieving quantum advantage. This is because currently available error-suppression methods for non-Clifford gates, such as magic-state distillation, require intensive resources that are far beyond the capabilities of current devices~\cite{rsa_2025}. 

For early fault-tolerant quantum computing, it is preferable to encode logical operations (especially non-Clifford operations) using relatively simple circuits and small stabilizer codes; these may not be completely fault-tolerant but can still reduce noise effectively. Much recent work on applying stabilizer codes to near-term quantum algorithm are based on the $[[n, n-2, 2]]$ quantum error detecting code (QEDC). For quantum chemistry, Ref.~\cite{422+H2} used the $[[4,2,2]]$ code (plus one ancilla) to protect a variational Unitary Coupled Cluster (UCC) ansatz~\cite{ucc_qc} and estimated the $\rm H_2$ molecule ground state energy. The logical error level of this ansatz was further investigated in Ref.~\cite{422+H2Meena}. Also, Ref.~\cite{QPE+iceberg} adopted Bayesian quantum phase estimation with QEDC to compute the $\rm H_2$ molecule ground state energy. For the quantum approximate optimization algorithm (QAOA), Ref.~\cite{QAOA+ErrorDetection} studied the performance of the $[[n, n-2, 2]]$ code family for the QAOA algorithm via numerical simulation, and the same team solved the low-autocorrelation binary sequences (LABS) problem with QAOA under the protection of an $[[n, n-2, 2]]$ code~\cite{LABS+iceberg}. For other quantum algorithms, Ref.~\cite{QITE+iceberg} successfully obtained both the ground and excited states of the spin singlet state with imaginary time evolution, where the $[[n, n-2, 2]]$ code effectively reduced errors in the circuit.

Another low-overhead quantum error correction architecture (the STAR architecture) for EFTQC applications, such as Trotter simulation and quantum phase estimation, uses the Clifford+$R_Z(\theta)$ universal gate set, where fault-tolerant Clifford operations are implemented by lattice surgery on different patches of surface code, and the logical rotations are realized by injecting ancilla states $|m_\theta\rangle = R_Z(\theta)|+\rangle$~\cite{STAR}. This architecture has an uncorrectable error rate of $O(p)$, or the better rate $O(|\theta|p)$ if the ancilla states are prepared with transversal multi-Pauli rotations~\cite{ImprovedSTAR}.

These attempts illustrate the continuing need for effective, low-overhead methods to reduce noise from non-Clifford gates for EFTQC applications. In this paper, we focus on encoded exponential operators for small quantum stabilizer codes. Exponential operators (or exponential maps, or multi-qubit rotations) are essential building blocks of many quantum algorithms that rely on the time evolution of a Hamiltonian system, or variational ans\"atze of this form. They take the form $\exp(-\ii\theta P)$, where $\ii$ is the imaginary unit, $\theta$ is an angle, and $P$ is an arbitrary $n$-qubit Pauli operator. An encoding of single- and two-qubit rotation gates for the $[[n, n-2, 2]]$ QEDC family was proposed in Ref.~\cite{iceberg}. Later, Ref.~\cite{gerhard2024weaklyFT} improved the logical single- and two-qubit rotations for QEDC codes by adding two controlled-NOT (CNOT) gates and one ancilla to remove all first-order errors except for analog errors resulting from rotation by an imprecise angle. We would like to extend the above ideas to more general exponential operators and other small stabilizer codes.

In Section~\ref{sec:encode}, we propose a method to encode exponential operators into arbitrary stabilizer codes. Note that this scheme---like those in the papers mentioned above---will not achieve complete fault-tolerance, so we also discuss optimizing to make the rate of logical errors as small as possible. In Section~\ref{sec:circuit}, we develop a systematic way to construct and search for circuits with small logical error. As examples, we apply our scheme to find optimal or near-optimal encoding circuits for logical single- and multi-qubit rotations in the $[[n, n-2, 2]]$ code, the five-qubit perfect code, the Steane code and the $[[15,7,3]]$ Hamming code in Sections~\ref{sec:qedc}, \ref{sec:5qcode}, \ref{sec:7qcode} and \ref{sec:15qcode}, respectively. We compare the performance of the resulting circuits to unencoded operations. Under the noise levels of current devices, our encoding scheme is 4--7 times less noisy, and at most $3\%$ of runs need to be discarded. The improvement grows to 10--30 times if the physical noise level of a single-qubit rotation decreases to $10^{-4}$. Finally, in Section~\ref{sec:discussion} we discuss the implications of our work.

\section{Encoding the Exponential Map}\label{sec:encode}
Encoding non-Clifford gates into a given $[[n, k, d]]$ stabilizer code with Clifford+T universal gate set and magic-state distillation requires a large amount of resources and has a large overhead. As a result, it cannot be applied to quantum algorithms in the near future due to the limited resources of current and near-term quantum computers. In this section, we study an encoding scheme specifically for the exponential map $\exp(-\ii \theta P)$, which is not fault-tolerant, but is useful for EFTQC applications and gives some protection from noise.

\subsection{Encoding Scheme and Verification}\label{sub:encode}
An encoding scheme that converts the logical exponential $\overline{U} = \overline{e^{-\ii\theta P}}$ into physical gates for an $[[n, k, d]]$ stabilizer code should satisfy the following two requirements:
\begin{enumerate}
\item The encoded operation $U$ defined on the $n$-qubit codespace is equivalent to the unencoded operation $\overline{U}$ defined on the $k$-qubit logical Hilbert space. In other words, for any logical state $|\overline{\psi}\rangle$ and its physical state $|\psi\rangle$, $U|\psi\rangle$ is the physical state of $\overline{U}|\overline{\psi}\rangle$. 

\item The encoded operator always transforms an encoded state into another state in the codespace. In other words, the encoded operator should commute with all stabilizer generators. 
\end{enumerate}
A general exponential map can be encoded into an arbitrary $[[n, k, d]]$ stabilizer code via
\begin{equation}\label{eq:encode_exp}
    \overline{e^{-\ii\theta P}} \xrightarrow{\text{ encoded into }} \exp(-\ii{\theta}P),
\end{equation}
where $\overline{e^{-\ii\theta P}} = \cos \theta I - \ii\sin\theta \overline{P} = e^{-\ii\theta\overline{P}}$ and $P$ is a physical representative of the logical Pauli $\overline{P}$.

To examine whether the construction in Eq.~\eqref{eq:encode_exp} meets the first requirement, consider an arbitrary logical state $|\overline{\psi}\rangle$ and its physical state $|\psi\rangle$. We have 
\begin{align} 
    e^{-\ii\theta\overline{P}}|\overline{\psi}\rangle =& \cos\theta|\overline{\psi}\rangle - \ii\sin\theta\overline{P}|\overline{\psi}\rangle \label{eq:prop1_eq1}, \\
    e^{-\ii\theta P}|\psi\rangle =& \cos\theta|\psi\rangle - \ii\sin\theta P|\psi\rangle \label{eq:prop1_eq2}.
\end{align}
Since the operator $P$ is the physical operator of $\overline{P}$, the state $P|\psi\rangle$ is the physical state of $\overline{P}|\overline{\psi}\rangle$. Therefore, $e^{-\ii\theta P} |\psi\rangle$ is the encoded state of $e^{-\ii\theta\overline{P}} |\overline{\psi}\rangle$.

The second requirement is also satisfied by the construction in Eq.~\eqref{eq:encode_exp}. Given that $P$ is a physical operator of $\overline{P}$, it holds for any stabilizer $S_j$ that
\begin{equation}
    \left[S_j , P \right]  = 0\iff S_j P = P S_j \iff S_j P S_{j}^{\dagger} = P.
\end{equation}
Since $Ve^{cA}V^{\dagger} = e^{cVAV^{\dagger}}$ for any Hermitian $A$, unitary $V$ and constant $c\in\mathbb{C}$, we have  
\begin{equation}
    S_je^{-i\theta P} S_j^{\dagger} = e^{-i\theta  S_j PS_j^{\dagger} } = e^{-i\theta P},
\end{equation}
and thus $S_je^{-i\theta P} =  e^{-i\theta P}S_j$. A stabilizer generator $S_j$ for the input state $|\psi\rangle$ is still a stabilizer generator for the output state $U|\psi\rangle$, because  $S_j U|\psi\rangle = US_j|\psi\rangle = U|\psi\rangle$. Thus, the input and output states are in the same codespace.

\begin{figure}
    \centering
    \includegraphics[width=0.45\textwidth]{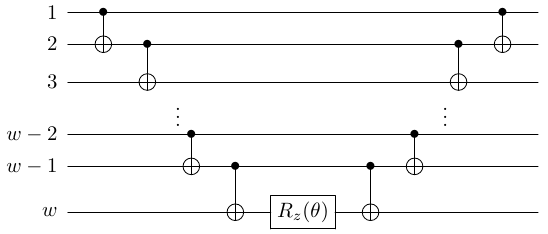}
    \caption{\label{fig:staircase} CNOT ladder implementation of $\exp(-\ii\theta Z^{\otimes w}/2)$. }
\end{figure}

\subsection{Optimizing the Circuit}\label{sub:optimal}
Let us consider a physical circuit implementing the exponential map with a total of $N$ noisy gates. The probability that a logical error occurs is given by the sum over $m$ of the $m$-th order logical error terms, where the $m$-th order logical error is
\begin{equation}\label{eq:m_order}
    {\rm Pr}(m\text{ faults})\,{\rm Pr}(\text{logical error}|m\text{ faults}).
\end{equation}
Here, a fault is a physical error arising from a single faulty gate. A faulty gate can be modeled as
\begin{equation}
\widetilde{\mathcal{U}} = \mathcal{N} \circ \mathcal{U} ,
\end{equation}
where $\mathcal{U}(\rho) = U\rho U^{\dagger}$ is a noise-free unitary channel and $\mathcal{N}$ is the noise channel associated with the gate. The logical error rate given $m$ faults depends on the precise error model. For simplicity, in this paper we will only consider $\mathcal{N}$ to be a Pauli channel for Clifford gates.

A quantum circuit is first-order fault tolerant if and only if ${\rm Pr}(\text{logical error}|m=1\text{ fault}) = 0$. Any circuit that directly implements the encoded operator $e^{-\ii\theta P}$ using Clifford$+R(\theta)$ as a universal gate set is not first-order fault tolerant, since an over- or under-rotation of the encoded operator by an angle $\Delta\theta$ will turn into a logical error $e^{-\ii\Delta \theta P}$. Given this limitation, circuits that minimize the first-order logical error for the encoded exponential map (referred to as optimal circuits) should ensure that logical errors arise only from such imprecise rotations, while all other physical errors are reliably detected by the syndrome measurement process.

\section{Circuit Implementation}\label{sec:circuit}
A widely used realization of an arbitrary exponential operator $\exp(-\ii\theta P)$ in the literature is a circuit of the form $Ue^{-\ii\theta P_1} U^{\dagger}$, where $e^{-\ii\theta P_1}$ is a single-qubit rotation with Pauli operator $P_1$ and $U$ and $U^{\dagger}$ are Clifford operators such that $UP_1 U^{\dagger} = P$. For a weight-$t$ Pauli operator $P$, this type of circuit usually contains $2(t-1)$ two-qubit Clifford gates and no more than $2t$ single-qubit Clifford gates. Examples of such implementations can be seen in Figure~4.19 of Ref.~\cite{nielsen_chuang} and in the CNOT ladder used in many papers (see e.g., Ref.~\cite{first_cnot_ladder} and Figure~\ref{fig:staircase}). This implementation exactly realizes the encoded exponential using Clifford gates plus a single-qubit rotation $R(\theta)$, rather than approximating it using the Clifford+T gate set. But it will not be fault-tolerant in general. In this paper, we will only study the level of logical noise and the construction of optimal or near-optimal encoded circuits with a single-qubit rotation gate for different stabilizer codes. There are other approaches to realize the exponential map, such as the implementation in Ref.~\cite{controll_Ry} using controlled rotation gates, and they may be worth exploring in future works. 

\begin{figure}
    \centering
    \includegraphics[width=0.45\textwidth]{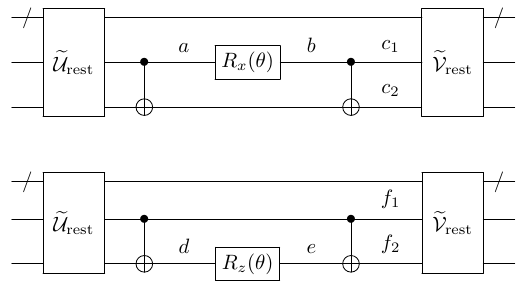}
    \caption{\label{fig:logical_error_location} Upper: an analog $X$ error can arise at location $b$. An $X$ error at location $a$ or $b$ is equivalent to an $X\otimes X$ error at locations $c_1, c_2$, and both will lead to the same logical error. Lower: an analog $Z$ error can arise at location $e$. A $Z$ error at location $d$ or $e$ and is equivalent to a $Z\otimes Z$ error at locations $f_1, f_2$, and both will lead to the same logical error. $\widetilde{\mathcal{U}}_{\rm rest}$ and $ \widetilde{\mathcal{V}}_{\rm rest}$ represent all other noisy Clifford gates in the circuit. }
\end{figure}

\subsection{Error Rate Calculation}\label{sub:error_rate}
In Section~\ref{sub:optimal}, we discuss imprecise rotations as a source of logical noise. An imprecise physical single-qubit rotation gate leads to the under- or over-rotation of the encoded rotation operator---that is, an analog error~\cite{gerhard2024weaklyFT}. Treating the error $\delta\theta$ in the rotation angle as a random variable with expectation $\mathbb{E}[\delta\theta] = 0$ and variance $\mathbb{E}[\delta\theta^2] = \sigma^2$, the analog error $\exp[-\ii (\delta\theta) P_1]$ can be modeled as a simple Pauli error process $P_1$ with error rate $q=\sigma^2$ occurring after the (perfect) single-qubit rotation~\cite{gerhard2024weaklyFT}. 

Suppose that an encoded exponential map $\exp(-\ii \theta P)$ with a weight-$t$ Pauli operator is implemented in the form of $Ue^{-\ii\theta P_1} U^{\dagger}$. If there are a total of $l$ physical Pauli errors across all two-qubit gates that will cause an undetectable logical error, the first-order logical error rate $p_L$ (not conditional on postselection) is approximately given by
\begin{equation}\label{eq:logical_error_rate}
    p_L = q + \frac{l}{15}p.
\end{equation}
In this calculation, we include analog errors (with error rate $q$) for the single-qubit rotation and neglect errors from single-qubit Clifford gates (which should all either be detectable or be equivalent to an error from a neighboring two-qubit gate). We assumed depolarizing noise for all two-qubit gates with error rate $p$. Other types of physical error may contribute to $p_L$ in a real quantum device, but for simplicity we will adopt this noise model for the remainder of this paper.

We observe that there are two physical errors from Clifford gates whose effect is equivalent to an analog error and will therefore also produce a logical error. The first is the physical error $P_1$ from the gate before the single-qubit rotation. This error commutes with the rotation. The second physical Pauli error is $VP_1V^{\dagger}$ from the gate $V$ after the rotation, which also results when error $P_1$ propagates through $V$. Figure~\ref{fig:logical_error_location} gives examples of these two cases as an illustration. From this result, the first-order logical error rate $p_L$ of an encoded exponential operator is at least
\begin{equation}\label{eq:ft_limit}
    p_{L} \geq q + \frac{2}{15}p.
\end{equation}
In the following sections, we show that this lower bound can be achieved by some encoded circuits in the form of $Ue^{-\ii\theta P_1} U^{\dagger}$ with postselection. An error correction process cannot completely fix logical errors from under- or over-rotation, and possible decoding errors may also introduce extra logical errors to the circuit. 

\subsection{Search for Optimal or Near-Optimal Circuits}\label{sub:search}
Circuits with different structures can result in different levels of logical noise. To find circuits with the best performance, one needs to generate a set of circuits that implement the same operator and select those with minimum $p_L$. For a given encoded exponential map $\exp(-\ii\theta P)$ with a weight-$t$ Pauli operator, we developed Algorithm~\ref{algo:candidate} below to generate a set of encoded circuits (defined as the {\it candidate set}) with $n_a$ ancilla qubits. 

\begin{algorithm}\label{algo:candidate}
    \caption{Candidate Set Construction}

    \SetKwInOut{Input}{input}
    \SetKwInOut{Output}{output}

    \SetKwData{BaseCircuits}{base\_circuits}
    \SetKwData{CandidateSet}{candidate\_set}
    \SetKwData{Circuit}{circuit}
    \SetKwFunction{XBase}{XBase}
    \SetKwFunction{YBase}{YBase}
    \SetKwFunction{ZBase}{ZBase}
    \SetKwFunction{AddAncilla}{AddAncilla}
    \SetKwFunction{AddSingleQubitGate}{Add1QGate}

    \Input{Pauli $P$ in encoded operator $\exp(-\ii \theta P)$, number of ancilla $n_a \geq 0$}
    \Output{\CandidateSet}
    \BlankLine

    $t \leftarrow$ weight of $P$\;
    $N \leftarrow t + n_a$ \;
    \BaseCircuits $\leftarrow $ \XBase{$N$}+\YBase{$N$}+\ZBase{$N$} \tcp*{Base circuit built on $N$ qubits with placeholders for $n_a$ ancillas} 
    \For{$U_c \in $ \BaseCircuits}{
        \eIf{$n_a > 0$}{
            $J \leftarrow$ all possible ancilia positions 
            \tcp*{Totally $\binom{N}{n_a}$ choices}
            \For{$j \in J$}
            {
                \Circuit $\leftarrow$ \AddAncilla{$U_c$, $j$} \;
                \Circuit $\leftarrow$ \AddSingleQubitGate{\Circuit, $P$, $j$}\;
                Append \Circuit to \CandidateSet \;
            }
        }{
        \Circuit $\leftarrow$ \AddSingleQubitGate{$U_c$, $P$}\;
        Append \Circuit to \CandidateSet \;
        }
    }
\end{algorithm}

The first step of the above algorithm is to generate a set of $N$-qubit base circuits, where $N = t + n_a$. The $X$-base, $Y$-base, and $Z$-base circuits are in the form $U e^{-\ii\theta P_1} U^{\dagger}$, where the single-qubit rotation gate is $R_X(\theta)$, $R_Y(\theta)$, or $R_Z(\theta)$, respectively. In the second step, we convert the $j$-th qubit in each base circuit into an ancilla qubit, which gives $\binom{N}{n_a}$ circuits for each base circuit. Since ancilla qubits may help detect more errors, this step expands the candidate set by including circuits that may approach or achieve the lower bound in Eq.~\eqref{eq:ft_limit}. Finally, we sandwich the above circuits between suitable single-qubit gates to convert them into candidate circuits for the desired $\exp(-\ii \theta P)$. Details on the construction of the base circuits ($\texttt{XBase}, \texttt{YBase}, \texttt{ZBase}$), the procedure for converting qubits into ancillas ($\texttt{AddAncilla}$) and sandwiching circuits with single-qubit gates ($\texttt{Add1QGate}$) can be found in Appendix~\ref{app:type_i}. 

% The operator $U$ in each base circuit is a sequence of $N-1$ controlled-NOT gates acting on different qubits. a diverse pool

%%%%%%%%%%%%%%%%%%%%%%%%%%%%%%%%%%%%%%%%%%%%%%%%%%%%%%%%%%%%%%%%%%%%%%%%%%%%%%%%

Now we can evaluate the logical error rate for each candidate and identify those with the minimum value of $l$ in Eq.~\eqref{eq:logical_error_rate} using Algorithm~\ref{algo:search} below.

\begin{algorithm}\label{algo:search}
    \caption{Best circuits in Candidate Set}

    \SetKwInOut{Input}{input}
    \SetKwInOut{Output}{output}

    \SetKwData{CandidateSet}{candidate\_set}
    \SetKwData{BestCircuit}{best\_circuit}
    \SetKwFunction{LogicalError}{LogicalError}
    
    \Input{\CandidateSet, stabilizer generators $\{S_j\}$ for $[[n,k,d]]$ code, ancilla generators $\{S^a_j\}$}
    \Output{\BestCircuit, i.e., circuits with minimum $l$}
    \BlankLine

    $l\leftarrow \infty$\;
    \For{$U_c \in $ \CandidateSet}{
        $l' \leftarrow $ \LogicalError{$U_c, \{S_j\}, \{S^a_j\}$}\;
        \uIf{$l' < l$}{
            $l \leftarrow l'$\;
            Clear \BestCircuit and append $U_c$\;
        }
        \ElseIf{$l' = l$}{
            Append $U_c$ to \BestCircuit \;
        }
    }
\end{algorithm}

In this algorithm, \LogicalError{$U_c, \{S_j\}, \{S^a_j\}$} evaluates $l$ in Eq.~\eqref{eq:logical_error_rate} for each candidate circuit under depolarizing noise. Specifically, it uses a brute-force approach to consider the $15$ possible Pauli errors that can occur after each two-qubit gate in the circuit and count how many of these errors lead to logical errors at the output, which gives the value of $l$. Circuits with the smallest $l$ are selected as the best candidates.

It is important to note that the search process in Algorithm~\ref{algo:search} does not guarantee that the selected circuit is strictly optimal (i.e., achieving $l=2$), as the outcome depends on the choice of the candidate set, the number of ancillas $n_a$, and the properties of the stabilizer code. We refer to circuits with the lowest $l\geq 2$ found among all candidates for a given stabilizer code as near-optimal circuits, and they can still offer significant advantages in practice. In the following sections, we present a detailed analysis of the best circuits for different stabilizer codes identified from the list of candidates.

\section{Encoded Exponential Map in Error Detecting Codes}\label{sec:qedc}
In this section, we mainly focus on the analysis of optimal and near-optimal circuits for encoded exponential operators we found for the $[[n, n-2, 2]]$ quantum error-detecting code (QEDC). This stabilizer code utilizes an even number of physical qubits $n$ to encode $k = n-2$ logical qubits with code distance $2$, and it is defined by two nonlocal Pauli operators, $X^{\otimes n}$ and $Z^{\otimes n}$. We choose a set of logical Pauli operators for logical qubits $j =1, \dotsc, n-2$ as $\overline{X}_j = X_jX_{n-1}$, $\overline{Z}_j = Z_jZ_{n}$ and $\overline{Y}_j = \ii\overline{X}_{j}\overline{Z}_j = Y_j X_{n-1} Z_n$. 

\begin{figure*}
    \centering
    \includegraphics[width=0.8\textwidth]{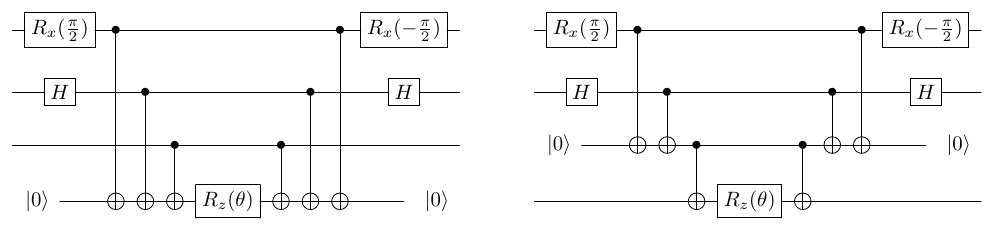}
\caption{\label{fig:logical_ry} Two near-optimal encoded circuits for $\overline{R_{Y_{j}}(\theta)} = \exp(-\ii\theta Y_{j}X_{n-1}Z_{n}/2)$ in the $[[n, n-2, 2]]$ code. For clarity, only the relevant qubits are shown in the circuit diagrams.}
\end{figure*}

\subsection{Single- and Two-qubit Rotations}
The logical single-qubit rotations in the $[[n, n-2, 2]]$ code are
\begin{equation}
    \begin{split}
    \overline{R_{X_j}(\theta)} =& \exp(-\ii\theta\overline{X}_j /2) = \exp(-\ii\theta X_jX_{n-1}/2),\\
    \overline{R_{Y_j}(\theta)} =& \exp(-\ii\theta\overline{Y}_j /2) = \exp(-\ii\theta Y_j X_{n-1}Z_n/2), \\
    \overline{R_{Z_j}(\theta)} =& \exp(-\ii\theta\overline{Z}_j /2) = \exp(-\ii\theta Z_jZ_{n}/2).
    \end{split}
\end{equation}
For $\overline{R_{X_j}(\theta)}$ and $\overline{R_{Z_j}(\theta)}$, we first consider circuits with no ancilla qubits and generate a candidate set of $4$ circuits for each operator. Some of these circuits are shown in Ref.~\cite{iceberg}. After running Algorithm~\ref{algo:search}, we find $l = 6$ for all $4$ circuits. We then build up another candidate set using one ancilla qubit in the circuit and search for the best, which returns a candidate set of $30$ circuits where $10$ of them reach the limit $l=2$. Two of the best candidates match the weakly fault-tolerant constructions of Ref.~\cite{gerhard2024weaklyFT}.

For $\overline{R_{Y_j}(\theta)}$, all $10$ circuits with no ancilla in the candidate set have $l=12$, while $14$ of $144$ circuits with one ancilla have $l=4$, which is worse than the $l=2$ limit. We extended our search to circuits with two ancilla qubits, but the best circuits we found have $l=4$. Thus, we consider $14$ circuits with one ancilla with $l=4$ as the near-optimal encoded circuits for $\overline{R_{Y_j}}(\theta)$ under the $[[n, n-2, 2]]$ code. Two of these built on $Z$-base circuits are shown in Figure~\ref{fig:logical_ry}.

We can also construct logical two-qubit rotation gates in the $[[n, n-2, 2]]$ code:
\begin{equation}
    \begin{split}
    \overline{R_{X_iX_j}(\theta)} =& \exp(-\ii\theta X_iX_{j}/2),\\
    \overline{R_{Y_iY_j}(\theta)} =& \exp(-\ii\theta Y_i Y_{j}/2), \\
    \overline{R_{Z_iZ_j}(\theta)} =& \exp(-\ii\theta Z_iZ_{j}/2). 
    \end{split}
\end{equation}
These encoded two-qubit rotations share a similar structure with those for $\overline{R_{X_j}(\theta)}$ and $\overline{R_{Z_j}(\theta)}$. As a result, the optimal circuits for these operations also achieve the $l=2$ limit and have the same circuit structure as the optimal implementations of $\overline{R_{X_j}(\theta)}$ and $\overline{R_{Z_j}(\theta)}$.

It is worth noting that single-qubit encoded rotation does not outperform unencoded rotation, because the logical error rate $p_L$ is always larger than the analog error rate of single-qubit rotation $q$. For two-qubit rotations, on the other hand, implementing an unencoded operator requires $2$ CNOT gates and one rotation gate. Therefore, the error rate of an unencoded two-qubit rotation is given by $p_{U} = 2p+q$, which is always larger than for encoded circuits without ancilla ($l=6$) and optimal encoded circuits with one ancilla ($l=2$). 

%%%%%%%%%%%%%%%%%%%%%%%%%%%%%%%%%%%%%%%%%%%%%%%%%%%%%%%%%%%%%%%

\begin{figure*}
    \centering
    \includegraphics[width=0.7\textwidth]{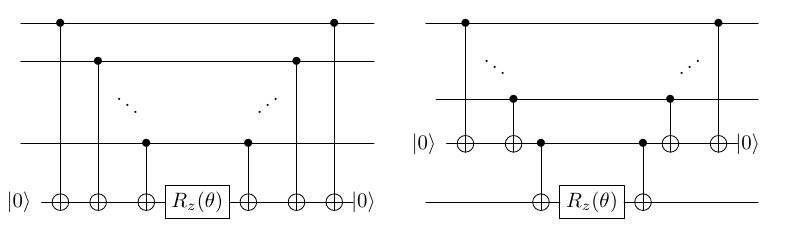}
\caption{\label{fig:multiq_rot_iceberg} Encoded circuits for $\overline{\exp(-\ii\theta Z^{\otimes w}/2)}$ in the $[[n, n-2, 2]]$ code. }
\end{figure*}

\begin{table*}[htp]
    \centering
    \caption{Details of Near-Optimal Circuit Search for Multi-qubit Rotation.}\label{tab:iceberg_multiq}
    \begin{tabular}{c@{\hspace{30pt}}c@{\hspace{30pt}}c@{\hspace{30pt}}c@{\hspace{30pt}}c@{\hspace{30pt}}c}
        \hline\hline
        % \multirow{2}{*}{Weight $w$} & Encoded Operator & No Ancilla & With Ancilla \\
        %   & Encoded Operator & No Ancilla & With Ancilla \\
        \multirow{2}{*}{Weight $w$} & \multirow{2}{*}{Encoded operator} & \multicolumn{2}{c}{Number of candidates} & \multicolumn{2}{c}{$l$ for best circuits} \\
        \cline{3-4}\cline{5-6}
        & & $n_a = 0$ & $n_a = 1$ & $n_a = 0$ & $n_a = 1$ \\
        \hline
        $w = 3,4$ & $\exp(-\ii\theta Z^{\otimes 4}/2)$ & 6    & 120    & 18 & 6\\
        $w = 5,6$ & $\exp(-\ii\theta Z^{\otimes 6}/2)$ & 120  & 5040   & 30 & 10\\
        $w = 7,8$ & $\exp(-\ii\theta Z^{\otimes 8}/2)$ & 5040 & 362880 & 42 & 14\\
        \hline
    \end{tabular}
\end{table*}

\subsection{Multi-qubit Rotations}\label{sub:multiq_iceberg}
We further investigate the level of logical noise for encoded multi-qubit rotations in the $[[n, n-2, 2]]$ code. Without loss of generality, we consider a rotation of a weight-$w$ logical $\overline{Z}$ operator, $\bigotimes_{j=1}^{w}\overline{Z}_j$, where $w$ ranges from $3$ to $8$. To keep the search space manageable, we only use $Z$-base circuits and set the number of ancillas $n_a$ to $0$ or $1$ to build up the candidate set. For these search parameters, we did not find an optimal circuit that reaches the lower limit $l=2$. Instead, we found two near-optimal circuits that have lower $l$ values than the other candidates (see Figure~\ref{fig:multiq_rot_iceberg}). The encoded operator with weight-$t$ Pauli $Z^{\otimes t}$, the number of candidate circuits and the minimum $l$ are reported in Table~\ref{tab:iceberg_multiq}. Note that an arbitrary encoded exponential can be constructed from a $Z$-base circuit by adding appropriate single-qubit gates before and after, and the impact of the physical noise from those single-qubit gates can be ignored, since any single fault will be detectable. Therefore, the rate of logical errors for multi-qubit $Z$-rotations is a good approximation of the rate of logical errors for an arbitrary encoded exponential of the same weight with circuits that have the same number of qubits and ancillas.

To investigate whether these encoded operators outperform unencoded operators, we first calculate the error rate of an unencoded multi-qubit rotation with a weight-$w$ Pauli operator $Z^{\otimes w}$. This requires $2(w-1)$ CNOT gates and one $R_Z(\theta)$ to implement the unencoded operator into the form $Ue^{-\ii\theta P_1}U^{\dagger}$, so the error rate is given by
\begin{equation}\label{eq:unencoded_multiq_error_rate}
    p_{U} = 2(w-1)p + q.
\end{equation} 
For comparison, we plot the error rates $p_U$ and $p_L$ for different multi-qubit rotations in Figure~\ref{fig:multi_iceberg}. Here, we only show the coefficient of $p$ in $p_L$ and $p_U$, which is $l/15$ for encoded circuits and $2(w-1)$ for unencoded circuits. The figure shows a significant suppression of noise when using encoded circuits, and the effect is much greater for larger-weight operators.

\begin{figure}
    \centering
    \includegraphics[width=0.45\textwidth]{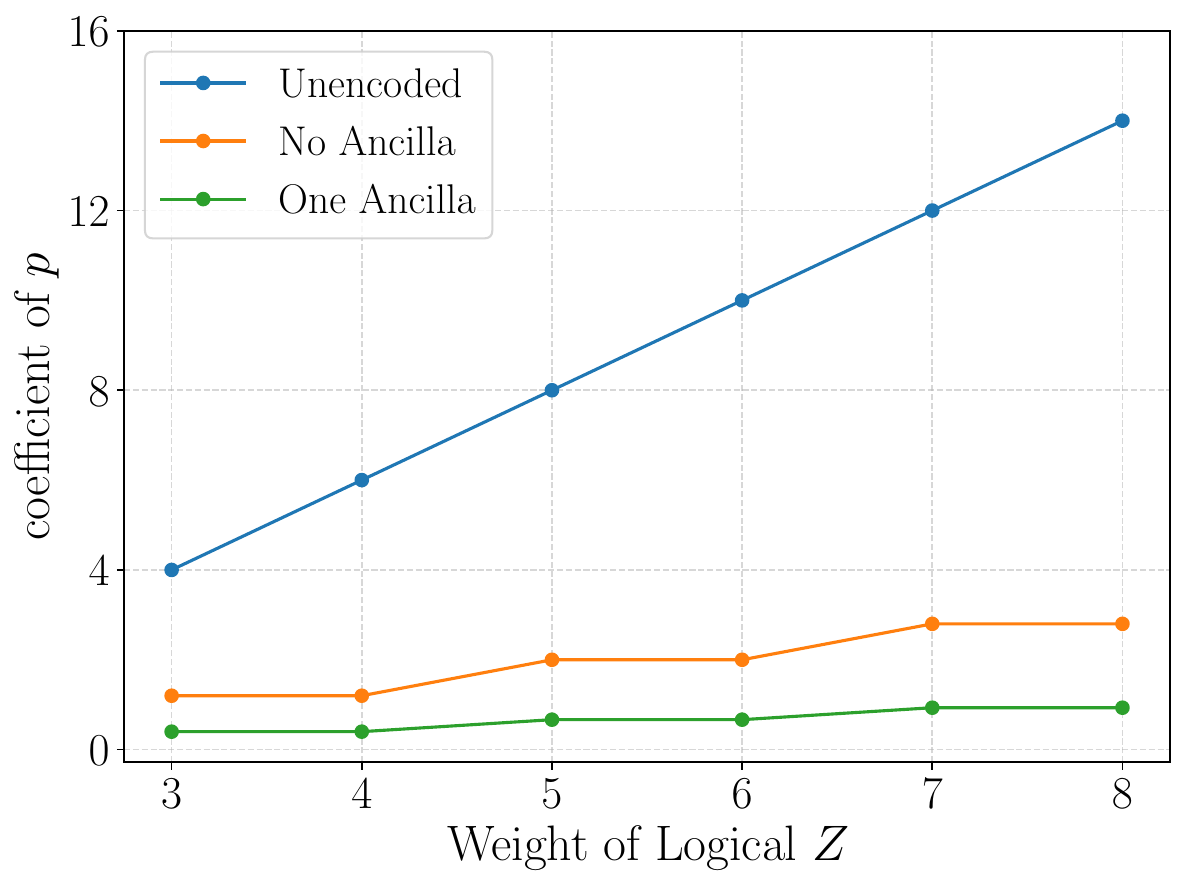}
\caption{\label{fig:multi_iceberg} Leading-order error rate for the near-optimal circuits we found for multi-qubit $Z$-operator rotations in the $[[n, n-2, 2]]$ code. Here, we plot the coefficient of $p$ in $p_L$. For the unencoded case, the $y$-axis shows $2(w-1)$. For encoded circuits, the $y$-axis shows $l/15$.}
\end{figure}

%%%%%%%%%%%%%%%%%%%%%%%%%%%%%%%%%%%%%%%%%%%%%%%%%%%%%%%%%%%%%%%

We note that the $[[n, n-2, 2]]$ code will reject the current run when at least one of the syndromes or the ancilla is measured to be $-1$. In this case, the upper bound of acceptance under postselection is given by
\begin{equation}\label{eq:post_selection}
    p_S \leq 1 - p_1 + p_L = 1 - 2(t+n_a - 1)p + \frac{l}{15}p,
\end{equation}
where $p_1 = 2(t+n_a - 1)p + q$ is the total first-order error rate of the encoded circuit and $n_a$ is the number of ancillas. The unencoded circuit does not discard any runs, so to make a fair comparison between the encoded and unencoded error rates we should consider the logical error rate in accepted data, $p_L / p_S$~\cite{gottesman2016smallFTexp}. For simplicity, we adopt $p = 0.001$ for the two-qubit error rate and calculate $p_U$, $p_S$, $p_L/p_S$ and the ratio between $p_U$ and $p_L/p_S$. Our calculation shows that with the rate of analog errors $q = 0.001$, the error rate for the encoded circuit with one ancilla is 5--7 times lower than the unencoded circuit for $w = 4, 6, 8$. If the analog error rate is further reduced to $q = 0.0001$, the encoded circuit will give 12--13 times less error than the unencoded circuit. These results also indicate better error suppression as $w$ increases. The details of these calculations are presented in Appendix~\ref{app:result}.

\section{\texorpdfstring{Encoded Exponential Map in the $[[5, 1, 3]]$ Code}{Encoded Exponential Map in the [[5, 1, 3]] Code}}\label{sec:5qcode}
\begin{figure*}
    \centering
    \includegraphics[width=0.7\textwidth]{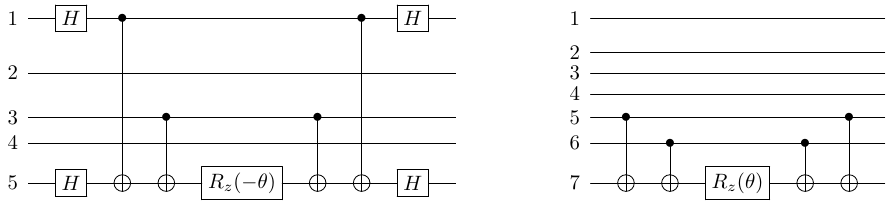}
    \caption{\label{fig:1q_rot_perfect_steane} Encoded circuits for $\overline{R_{Z}(\theta)}$ in the $[[5, 1, 3]]$ code (left) and the $[[7, 1, 3]]$ code (right). }
\end{figure*}

In this section, we study optimal circuits for encoded exponential operators in the $[[5,1,3]]$ perfect code. The stabilizer generators of this code are
\begin{equation}
S = \{ XZZXI, IXZZX, XIXZZ, ZZXIX \}.
\end{equation}
This code has transversal logical Pauli operators $\overline{X} = X^{\otimes 5}$, $\overline{Y} = Y^{\otimes 5}$ and $\overline{Z} = Z^{\otimes 5}$. Multiplying these logical operators by stabilizer generators gives another (minimum-weight) set of logical Pauli operators:
\begin{equation}
    \overline{X} = -YIXIY, \overline{Y} = -ZIYIZ, \overline{Z} = -XIZIX.
\end{equation}

\subsection{Single-Qubit Rotations}
The logical single-qubit rotations in the $[[5, 1, 3]]$ code are
\begin{equation}
    \begin{split}
    \overline{R_{X}(\theta)} =& \exp(\ii\theta Y_1X_3Y_5/2),\\
    \overline{R_{Y}(\theta)} =& \exp(\ii\theta Z_1Y_3Z_5/2), \\
    \overline{R_{Z}(\theta)} =& \exp(\ii\theta X_1Z_3X_5/2).
    \end{split}
\end{equation}
We generated a set of $10$ candidate circuits with no ancilla qubits and found that all of them reach the limit $l=2$. We note that when running Algorithm~\ref{algo:search} we do not include any type of error correction procedure. Instead, we assume only that postselection is done at the end of the circuit to discard runs with detectable errors, just as in the case of the $[[n,n-2,2]]$ codes, and only $l=2$ physical errors from CNOT gates cannot be removed by this process. 

\begin{figure*}
    \centering
    \includegraphics[width=0.7\textwidth]{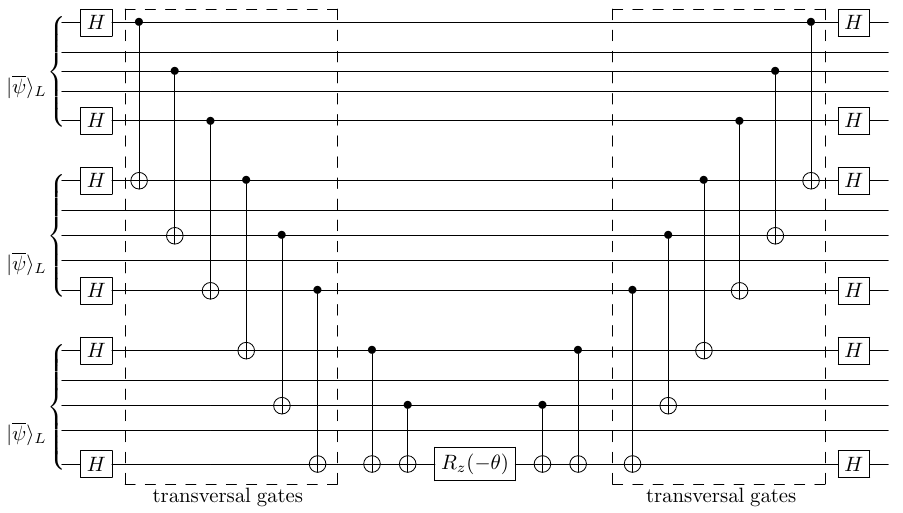}
\caption{\label{fig:multiq_rot_perfect} Encoded circuits for $\overline{\exp(-\ii\theta Z^{\otimes 3}/2)}$ on three blocks of the $[[5, 1, 3]]$ code.}
\end{figure*}

There are two main advantages of using postselection, even with quantum codes that can correct errors. First, it is more practical to realize postselection in a near-term quantum computer, since it does not require the rapid feedback process needed to perform error correction. Second, some physical errors cannot be completely removed by error correction, which in the absence of full fault-tolerance can even introduce logical errors. For example, a $Z_1$ error before the $\overline{R_{X}(\theta)}$ of the $[[5, 1, 3]]$ code will propagate through the encoding circuit and become another error $Z_1 \overline{R_{X}(-2\theta)}$, where only the $Z_1$ part can be corrected while the extra rotation will remain. However, postselection can eliminate both the $Z_1$ error and the extra rotation by rejecting the run when a stabilizer is measured to be $-1$. We will therefore only analyze postselection in this section.

\subsection{Multi-Qubit Rotations}\label{sub:multiq_d3}
We now investigate multi-qubit rotations on several blocks of logical qubits protected by the $[[5,1,3]]$ code. Consider a rotation of weight-$w$ logical $\overline{Z}^{\otimes w}$ again, where
\begin{equation}
    \overline{Z}^{\otimes w} = \bigotimes_{i=1}^{w} (-1)^{w}X_{5i-4}Z_{5i-2}X_{5i}.
\end{equation}
We can construct specific circuits with transversal CNOT gates to implement the multi-qubit rotation. An example of a $w=3$ circuit is shown in Figure~\ref{fig:multiq_rot_perfect}. It can be easily verified that this type of implementation has $l = 2$ under postselection after the entire circuit: the transversal parts on both sides of the circuit are naturally first-order fault-tolerant, and logical (analog) errors only come from the circuit between these parts, which is very similar to a single-qubit rotation on the last qubit with $l=2$. Note that the transversal physical CNOTs in Figure~\ref{fig:multiq_rot_perfect} are not valid logical gates for the 5-qubit code, but the circuit as a whole is a valid logical multi-qubit rotation and preserves the codespace according to Section~\ref{sub:encode}.

To understand the performance of the encoded circuit, we calculate the error rates as in Section~\ref{sub:multiq_iceberg} with $p=0.001$. The result shows a reduction of 2--7 times for $q = 0.001$ to as much as 33 times for $q=0.0001$. We also observe greater improvement when $w$ increases, similar to the results in Section~\ref{sub:multiq_iceberg}. Our simple calculation also reduces concerns about the postselection rate, since fewer than $3\%$ of runs are discarded in our sample at the physical noise level of current devices. We refer readers to Appendix~\ref{app:result} for more details.

We can also compare our construction to a different encoding scheme for the exponential operation. Ref.~\cite{5_1_3_rotation} proposed a transversal Hadamard gate for $\overline{H}$ and a ``pieceably fault-tolerant'' $\overline{\rm  CZ}$ gate implementation. Here, the latter consists of two subcircuits (``pieces'') with an intermediate error correction between the two pieces, which is denoted by $\overline{\rm  CZ} = C_b\cdot {\rm EC}\cdot C_a$. If a pieceably fault-tolerant CNOT gate can be built from $\overline{\rm CZ}$ and $\overline{H}$ gates, the encoded operator $\overline{\exp(-\ii\theta Z^{\otimes w})}$ can be implemented by logical $\overline{\rm CNOT}$ gates between neighboring blocks of qubits, together with an encoded single-qubit rotation on the last logical qubit. In this case, the resulting circuit will be more complicated but still have $l=2$. Therefore, our encoding scheme offers a simple but effective way to implement the logical operator.

\section{\texorpdfstring{Encoded Exponential Map in the $[[7, 1, 3]]$ Code}{Encoded Exponential Map in the [[7, 1, 3]] Code}}\label{sec:7qcode}
\begin{figure*}
    \centering
    \includegraphics[width=0.7\textwidth]{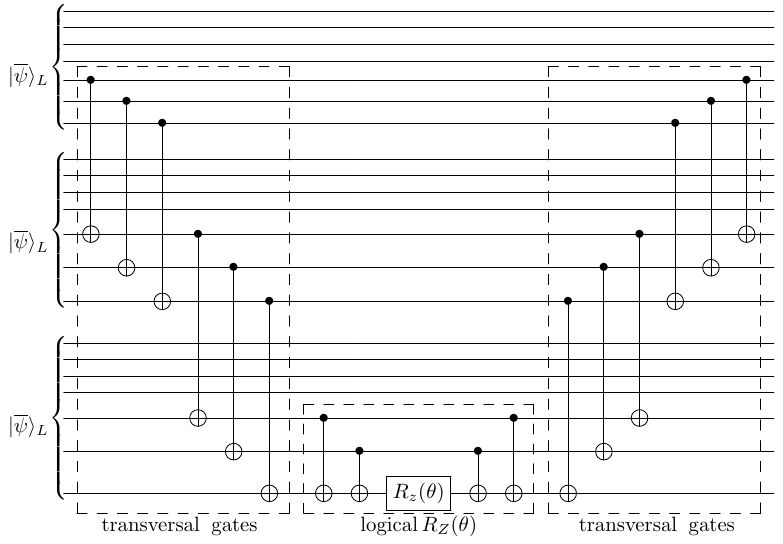}
\caption{\label{fig:multiq_rot_steane} Encoded circuits for $\overline{\exp(-\ii\theta Z^{\otimes 3}/2)}$ under three blocks of the $[[7, 1, 3]]$ code. }
\end{figure*}

This section mainly focuses on finding optimized encoded exponential operators for the $[[7,1,3]]$ code. The stabilizer generators of this code are given by
\begin{equation}
    \begin{split}
    S = \{ &X_1X_2X_3X_4, X_2X_3X_5X_6, X_3X_4X_6X_7, \\
           &Z_1Z_2Z_3Z_4, Z_2Z_3Z_5Z_6, Z_3Z_4Z_6Z_7\}.
\end{split}
\end{equation}

This code has logical Pauli operators, $\overline{X} = X^{\otimes 7}$, $\overline{Y} = Y^{\otimes 7}$ and $\overline{Z} = Z^{\otimes 7}$. Multiplying these logical operators by stabilizer generators gives equivalent weight-3 operators:
\begin{equation}
    \overline{X} = IIIIXXX, \overline{Y} =IIIIYYY, \overline{Z} = IIIIZZZ.
\end{equation}
The logical single-qubit rotations in the $[[7, 1, 3]]$ code are
\begin{equation}
    \begin{split}
    \overline{R_{X}(\theta)} =& \exp(-\ii\theta X_5X_6X_7/2),\\
    \overline{R_{Y}(\theta)} =& \exp(-\ii\theta Y_5Y_6Y_7/2), \\
    \overline{R_{Z}(\theta)} =& \exp(-\ii\theta Z_5Z_6Z_7/2).
    \end{split}
\end{equation}
We created a candidate set of $10$ circuits with no ancilla qubits for these logical operations, and find that all of them are optimal with $l=2$ under postselection.

For multi-qubit rotations, we adopt the same strategy as in the previous section and construct a circuit from transversal gates in Figure~\ref{fig:multiq_rot_steane}. This circuit gives $l=2$ under postselection. The resulting error rate and acceptance rate are the same as for the $[[5,1,3]]$ code in Section~\ref{sec:5qcode} (see Tables~\ref{tab:perfect_multiq_comparison_1} and \ref{tab:perfect_multiq_comparison_2} for details).

Note that, unlike the $[[5,1,3]]$ code, the Steane code has transversal CNOT gates, so circuits built with transversal CNOT gates are equivalent to using logical $\overline{\rm CNOT}$ and single-qubit rotations to produce an encoded version of the ``staircase'' CNOT ladders shown in Figure~\ref{fig:staircase}.

\section{\texorpdfstring{Encoded Exponential Map in the $[[15, 7, 3]]$ Code}{Encoded Exponential Map in the [[15, 7, 3]] Code}}\label{sec:15qcode}
The $[[15,7,3]]$ Hamming code, from the same code family as the $[[7,1,3]]$ code in the previous section, has four $X$ and four $Z$ stabilizers. Both sets of stabilizers use the following parity-check matrix~\cite{Chao2018}:
\begin{equation}
\left(
\begin{array}{c c c c c c c c c c c c c c c}
0&0&0&0&0&0&0&1&1&1&1&1&1&1&1\\
0&0&0&1&1&1&1&0&0&0&0&1&1&1&1\\
0&1&1&0&0&1&1&0&0&1&1&0&0&1&1\\
1&0&1&0&1&0&1&0&1&0&1&0&1&0&1
\end{array}\right) .
\end{equation}
Here, the leftmost column represents the first physical qubit and the rightmost column is for the $15$th qubit. For example, the first $X$-stabilizer is given by $S_1 = X_8X_9X_{10}X_{11}X_{12}X_{13}X_{14}$, and the first $Z$-stabilizer is given by $S_5 = Z_8Z_9Z_{10}Z_{11}Z_{12}Z_{13}Z_{14}$. Similarly, the logical Pauli operators can be defined based on the following seven weight-five binary strings~\cite{Chao2018}:
\begin{equation}
\left(
\begin{array}{c c c c c c c c c c c c c c c}
1&1&0&1&0&0&0&1&0&0&0&0&0&0&1\\
1&1&0&0&1&0&0&0&0&1&0&1&0&0&0\\
1&1&0&0&0&1&0&0&0&0&1&0&0&1&0\\
1&1&0&0&0&0&1&0&1&0&0&0&1&0&0\\
1&0&0&1&0&1&0&0&1&1&0&0&0&0&0\\
1&0&0&1&0&0&1&0&0&0&0&1&0&1&0\\
1&0&0&0&0&0&0&1&0&1&0&0&1&1&0
\end{array}\right) .
\end{equation}
For example, the first row gives $\overline{X}_1 = X_1X_2X_4X_8X_{15}$ and $ \overline{Z}_1 = Z_1Z_2Z_4Z_8Z_{15}$. The remaining strings specify the logical operators $\overline{X}_2, \overline{Z}_2$ through $\overline{X}_7, \overline{Z}_7$. Since the Hamming code has distance three, there exist weight-3 representations of the logical Pauli operators. However, for simplicity we will stick to the above definition of the single-qubit logical Pauli operators to construct single- or multi-qubit logical rotations.

Some examples of logical single-qubit rotations in the $[[15, 7, 3]]$ code are
\begin{equation}
    \begin{split}
    \overline{R_{X_1}(\theta)} =& \exp(-\ii\theta X_1X_2X_4X_8X_{15}/2),\\
    \overline{R_{Y_1}(\theta)} =& \exp(-\ii\theta Y_1Y_2Y_4Y_8Y_{15}/2), \\
    \overline{R_{Z_1}(\theta)} =& \exp(-\ii\theta Z_1Z_2Z_4Z_8Z_{15}/2).
    \end{split}
\end{equation}
These rotations apply to the first logical qubit. For $\overline{R_{Z_1}(\theta)}$, we searched through $24$ $Z$-based circuits and found that all of the circuits give the value $l=2$.

For multi-qubit logical rotations in the same block of $15$ physical qubits, we consider the following eight cases with different weights in the encoded operators:
\begin{equation}
    \begin{split}
        \text{weight-}4: &\overline{Z}_4\overline{Z}_5\overline{Z}_6\overline{Z}_7 = Z_{2,6,8,12} . \\
        \text{weight-}6: &\overline{Z}_1\overline{Z}_2 = Z_{4,5,8,10,12,15} .\\
        \text{weight-}7: &\overline{Z}_5\overline{Z}_6\overline{Z}_7 = Z_{1, 6, 7, 8, 9, 12, 13} .\\
        \text{weight-}8: &\overline{Z}_2\overline{Z}_3\overline{Z}_4\overline{Z}_5 = Z_{2, 4, 5, 7, 11, 12, 13, 14} .\\
        \text{weight-}9: &\overline{Z}_1\overline{Z}_2\overline{Z}_3\overline{Z}_4\overline{Z}_5 = Z_{1, 5, 7, 8, 11, 12,13, 14, 15} .\\
        \text{weight-}10: &\overline{Z}_1\overline{Z}_2\overline{Z}_3\overline{Z}_5\overline{Z}_6\overline{Z}_7 = Z_{2, 4, 5, 7, 9, 10, 11, 13, 14, 15} .\\
        \text{weight-}11: &\overline{Z}_1\overline{Z}_2\overline{Z}_3 = Z_{1, 2, 4, 5, 6, 8, 10, 11,12,14,15} .\\
        \text{weight-}12: &\overline{Z}_1\overline{Z}_2\overline{Z}_3\overline{Z}_4 = Z_{4, 5, 6, 7, 8, 9, 10,11,12,13,14,15} .
    \end{split}
\end{equation}
For compactness, we use here a shorthand notation for a product of Pauli operators: $Z_{i_1,\ldots,i_k} \equiv Z_{i_1} Z_{i_2} \cdots Z_{i_k}$. For all of these logical rotations with weight less than $12$, if we implement them with the staircase circuit shown in Fig.~\ref{fig:staircase}, they have $l=2$. For the weight-$12$ circuit, only a circuit similar to Fig.~\ref{fig:logical_ry} with no ancilla qubits reaches $l=2$. The error and acceptance rate calculations are given in Appendix~\ref{app:result}.

\section{Discussion}\label{sec:discussion}
In this paper, we have proposed a systematic encoding scheme for exponential operators that does not rely on decomposition into the Clifford+T gate set. As a proof of concept, we studied the $[[n, n-2, 2]]$, $[[5,1,3]]$, $[[7,1,3]]$, and $[[15,7,3]]$ codes and constructed encoding circuits for the exponential operator that minimize the logical error rate after postselection. While our encoding scheme cannot achieve first-order fault tolerance, our analysis suggests that its simplicity and the effectiveness in error reduction still offer practical advantages for protecting quantum algorithms from noise in the early fault-tolerant quantum computing (EFTQC) era.

We emphasize the importance of postselection when using our encoding scheme, even when the circuit is encoded into an error-correcting code. Postselection eliminates errors by discarding runs that are flagged by syndrome measurements to have errors in the circuit. This process is particularly effective in the context of exponential operators, where the error-correction procedure can induce extra logical rotations as errors that cannot easily be handled by error-correction process. It would be interesting to compare these results using postselection to results where errors are corrected by classical postprocessing; this comparison will wait for future work.

It is important to point out that under our noise model and using postselection, encoded multi-qubit rotations perform much better than unencoded circuits, and the improvement is more significant as the number of qubits increases. Moreover, multi-qubit rotation can be encoded using transversal gates between blocks of qubits, even if the transversal gate sequence itself is not a valid logical gate. 

We believe that magic state distillation remains a promising method to fault-tolerantly implement non-Clifford gates, and its potential application on quantum chemistry was discussed in Ref.~\cite{trout2015magic}. The use of magic states or other fault-tolerant methods will be necessary for full scalability of fault-tolerant quantum computing. However, our approach offers a route towards application on EFTQC algorithms in the near term due to its simplicity, compact circuits, and lack of real-time feedback. A detailed analysis of our scheme for specific quantum problems, both with simple postselection at the end and with mid-circuit syndrome measurements, will be done in the future to evaluate the ability of near-term encoding schemes to achieve quantum advantage.

\section*{Acknowledgment}
DZ and TAB thank Prithviraj Prabhu, Christopher Gerhard, and Zichang He for fruitful discussions. This work was supported in part by NSF Grants 1911089 and 2316713, and by the U.S. Army Research Office under contract number W911NF2310255.

\bibliography{reference}

\newpage
\onecolumngrid
\appendix
\setcounter{equation}{0}

% \section{Exponential Map in Quantum Algorithm}\label{app:exp}

\section{Constructing encoded circuits}\label{app:type_i}
To search for optimal circuits, we need to generate a set of candidate circuits that implement the same encoded exponential map $Ue^{-i\theta P_1}U^{\dagger}$ , and select those with the lowest $l$. In this section, we provide three circuit tricks that can be used to create candidate circuits. 

\subsection{Recursive Construction of Base Circuits}
A circuit for $\exp(-\ii \theta P')$ can be built by sandwiching a circuit for $\exp(-\ii\theta P)$ with controlled-NOT gates (as shown below), where $P$ is a weight-$t$ Pauli operator and $P'$ is a weight-$(t+1)$ Pauli operator. 

\begin{prop}[Sandwich by CNOTs] \label{prop:cnot_pair_general}
    A circuit for $\exp(-\ii \theta P')$ can be constructed from either of the following two structures,
    \begin{center}
        \includegraphics[width=0.6\textwidth]{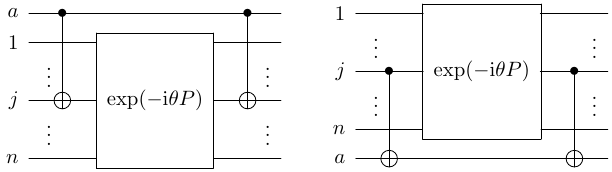}
    \end{center}
    if $P'$ and $P$ are related by
    \begin{equation}
        P'_{aj} = {\rm CNOT}_{aj}(I_a\otimes P_{j}){\rm CNOT}_{aj} \;\;\;\;\text{(using the left circuit),}
    \end{equation}
    or
    \begin{equation}
        P'_{ja} = {\rm CNOT}_{ja}(P_{j}\otimes I_a){\rm CNOT}_{ja}\;\;\;\;\text{(using the right circuit).}
    \end{equation}
    Here, $P_{j}$ is the component of the weight-$t$ Pauli operator $P$ on qubit $j$, and $P'_{aj} = P'_a \otimes P'_j$ is the component of weight-$(t+1)$ Pauli operator $P'$ on qubits $a$ and $j$. $P$ and $P'$ are identical on all the other qubits.
\end{prop} 
\begin{proof}
    The construction on the left can be written as  
    \begin{equation}\label{eq:general}
        U = {\rm CNOT}_{aj}[I_a\otimes \exp(-\ii\theta P)]{\rm CNOT}_{aj}.
    \end{equation}
    Here we ignore identities on other qubits for simplicity. Expanding this equation, we have
    \begin{equation}\label{eq:expand}
        U = {\rm CNOT}_{aj}[I_a\otimes \cos\theta I^{\otimes n}]{\rm CNOT}_{aj} - {\rm CNOT}_{aj}[I_a\otimes \ii\sin\theta \,P]{\rm CNOT}_{aj}.
    \end{equation}
    Notice that ${P}'_{aj} = {\rm CNOT}_{aj}[I_a\otimes P_{j}]{\rm CNOT}_{aj}$. For Eq.~\eqref{eq:expand} we have
    \begin{equation}
        U = \cos\theta[I_a\otimes  I^{\otimes n}] - i\sin\theta P' = \exp(-\ii\theta P').
    \end{equation}
    The construction on the right can also be proved by a similar procedure.
\end{proof}

For example, the circuit of $\exp(-\ii\theta Z^{\otimes t}/2)$ can be created from $\exp(-\ii\theta Z^{\otimes t-1}/2)$ using the left construction, while the circuit of $\exp(-\ii\theta X^{\otimes t}/2)$ can be created from $\exp(-\ii\theta X^{\otimes t-1}/2)$ using the right construction. By adding CNOT pairs and applying this trick recursively, one can build basic operations $\exp(-\ii\theta Z^{\otimes t}/2)$ and $\exp(-\ii\theta X^{\otimes t}/2)$ for any $t$ with one single-qubit rotation $R_Z(\theta)$ or $R_X(\theta)$, respectively. A different CNOT position will give a different circuit structure. We build base circuits $\texttt{XBase}, \texttt{YBase}, \texttt{ZBase}$ using above trick in Algorithm~\ref{algo:candidate} to generate candidate circuits for the encoded exponential map.  

% The operator $U$ in each base circuit is a sequence of $N-1$ controlled-NOT gates acting on different qubits. Second, we employ a \textbf{recursive construction} technique. This method builds a circuit for a weight-$(t+1)$ Pauli operator, $P'$, from a circuit for a weight-$t$ operator, $P$, by adding a pair of CNOT gates. By recursively applying this technique, we can construct circuits for basic operations of any weight, such as $\exp(-\ii\theta Z^{\otimes t})$ and $\exp(-\ii\theta X^{\otimes t})$. The placement of the CNOT pairs creates different circuit structures, expanding our set of candidates.

\subsection{Adding Ancillas}
Adding an ancilla should not change the exponential operator that a circuit implements, but it may reduce the level of logical error. The following circuit trick shows how to add an ancilla to a circuit for $\exp(-\ii\theta P)$ from the circuit of $\exp(-\ii\theta P')$.
\begin{prop}[Adding Ancilla] \label{prop:extra_ancilla}
    A weight-$(t+1)$ Pauli operator $P'$ can be written as
    \begin{equation}
        P'= P_{1} \otimes \cdots \otimes P_{j-1}\otimes P_{j}\otimes P_{j+1}\otimes \cdots \otimes P_{t+1}.
    \end{equation}
    If the weight-$t$ Pauli operator $P$ is written as 
    \begin{equation}
        P = P_{1} \otimes \cdots \otimes P_{j-1}\otimes P_{j+1}\otimes \cdots \otimes P_{t+1},
    \end{equation}
    then by setting the initial state of the $j$-th qubit as the $+1$ eigenstate of $P_j$ and postselecting the output to be $+1$ when measuring $P_j$ at the end, the circuit implementing $\exp(-\ii\theta P')$ becomes a circuit implementing $\exp(-\ii\theta P)$ with an ancilla.  
\end{prop}

\begin{proof}
    Expand $\exp(-\ii\theta P')$
    \begin{equation}
    \exp(-\ii\theta P') = \cos\theta I - \ii\sin\theta P'.
    \end{equation}
    If we choose the initial state $|\psi_j\rangle$ of the $j$-th qubit be the $+1$ eigenstate of $P_j$, and the (general) state of all the other qubits is $|\phi\rangle$, then we can rewrite
    \begin{equation}
    \left(\cos\theta I - \ii\sin\theta P' \right) \left( |\phi\rangle\otimes |\psi_j\rangle \right) = 
    \left(\cos\theta I - \ii\sin\theta P \right)|\phi\rangle \otimes |\psi_j\rangle ,
    \end{equation}
    where $P = P_1 \otimes P_2 \otimes \cdots \otimes P_{j-1}\otimes P_{j+1}\otimes\cdots \otimes P_{t+1}$ is a weight-$t$ Pauli operator related to $P'$ by replacing $P_j$ by $I$. We can postselect at the end of the circuit by measuring $P_j$ on qubit $j$ and checking if we get $+1$.
\end{proof}

In principle, one could add multiple ancilla qubits by applying this trick repeatedly, which might further decrease the likelihood of logical errors in the circuit. Adding multiple ancillas using above trick is used for $\texttt{AddAncilla}$ in Algorithm~\ref{algo:candidate} to generate candidate circuits for the encoded exponential map. %In this paper, we consider candidate circuits with at most one ancilla and search for the noise-resilient circuit of encoded exponential from candidates.

% Finally, we consider \textbf{adding an ancilla qubit}. While an ancilla does not alter the implemented exponential operator, it can reduce the logical error rate. This is achieved by preparing the ancilla in a specific state and post-selecting on its measurement outcome. 

\subsection{Adding Single-qubit Gates to Base Circuits}
To build a circuit for $\exp(-\ii\theta P')$ where $P'$ is an arbitrary weight-$t$ Pauli operator, a widely-used trick in the literature is to sandwich the base circuit for $\exp(-\ii\theta Z^{\otimes t})$ between single-qubit gates such as the Hadamard gate $H$ and rotation gates $R_{x}(\pm \pi/2)$ and $R_{z}(\pm \pi/2)$. These single-qubit gates transform $Z^{\otimes t}$ in the exponential into an arbitrary Pauli operator $P'$. We can choose the single-qubit gates based on the following rules:
\begin{itemize}
    \item Use the Hadamard gate to make $X$: $X = HZH$. 
    \item Use $R_{x}(\pm \pi/2)$ to make $Y$: $Y = R_{x}\left(-\pi/2\right)ZR_{x}\left(\pi/2\right)$.
\end{itemize}
We can use similar tricks to make $\exp(-\ii\theta P')$ starting from a circuit for $\exp(-\ii\theta X^{\otimes t})$ and sandwiching between Hadamard or $R_z(\pm\pi/2)$ gates on the different qubits.

This method is used for $\texttt{Add1QGate}$ in Algorithm~\ref{algo:candidate} to generate candidate circuits of encoded exponential map.

\section{Error Rate Calculation Result}\label{app:result}
\begin{table*}[htp]
    \centering
    \caption{Error rates for $\overline{\exp(-\ii\theta Z^{\otimes w}/2)}$ in the $[[n, n-2, 2]]$ code with $q = 0.001$. }\label{tab:iceberg_multiq_comparison_1}
    \begin{tabular}{c@{\hspace{15pt}}c@{\hspace{15pt}}c@{\hspace{15pt}}c@{\hspace{15pt}}c@{\hspace{25pt}}c@{\hspace{25pt}}c@{\hspace{25pt}}c}
        \hline\hline
        % \multirow{2}{*}{Weight $w$} & Encoded Operator & No Ancilla & With Ancilla \\
        %   & Encoded Operator & No Ancilla & With Ancilla \\
        \multirow{2}{*}{Weight} &\multirow{2}{*}{$p_U\,\,(\times 10^{-3})$} & \multicolumn{2}{c}{$p_S$ (\%)} & \multicolumn{2}{c}{$p_L/p_S\,\,(\times 10^{-3})$} & \multicolumn{2}{c}{$p_U / (p_L/p_S)$} \\
        \cline{3-4}\cline{5-6}\cline{7-8}
        & & $n_a = 0$ & $n_a = 1$ & $n_a = 0$ & $n_a = 1$& $n_a = 0$ & $n_a = 1$ \\
        \hline
        $w = t = 4$ & 7  &  99.52 & 99.24 & 2.21 & 1.41 & 3.17 & 4.96\\
        $w = t = 6$ & 11 &  99.20 & 98.87 & 3.02 & 1.69 & 3.64 & 6.53\\
        $w = t = 8$ & 15 &  98.88 & 98.50 & 3.84 & 1.96 & 3.90 & 7.64\\
        \hline
    \end{tabular}
\end{table*}

\begin{table*}[htp]
    \centering
    \caption{Error rates for $\overline{\exp(-\ii\theta Z^{\otimes w}/2)}$ in the $[[n, n-2, 2]]$ code with $q = 0.0001$}\label{tab:iceberg_multiq_comparison_2}
    \begin{tabular}{c@{\hspace{15pt}}c@{\hspace{15pt}}c@{\hspace{15pt}}c@{\hspace{15pt}}c@{\hspace{25pt}}c@{\hspace{25pt}}c@{\hspace{25pt}}c}
        \hline\hline
        % \multirow{2}{*}{Weight $w$} & Encoded Operator & No Ancilla & With Ancilla \\
        %   & Encoded Operator & No Ancilla & With Ancilla \\
        \multirow{2}{*}{Weight} &\multirow{2}{*}{$p_U\,\,(\times 10^{-3})$} & \multicolumn{2}{c}{$p_S$ (\%)} & \multicolumn{2}{c}{$p_L/p_S\,\,(\times 10^{-3})$} & \multicolumn{2}{c}{$p_U / (p_L/p_S)$} \\
        \cline{3-4}\cline{5-6}\cline{7-8}
        & & $n_a = 0$ & $n_a = 1$ & $n_a = 0$ & $n_a = 1$& $n_a = 0$ & $n_a = 1$ \\
        \hline
        $w = t = 4$ & 6.1  &  99.52 & 99.24 & 1.31 & 0.50 & 4.67 & 12.11\\
        $w = t = 6$ & 10.1 &  99.20 & 98.87 & 2.12 & 0.78 & 4.77 & 13.02\\
        $w = t = 8$ & 14.1 &  98.88 & 98.50 & 2.93 & 1.05 & 4.81 & 13.44\\
        \hline
    \end{tabular}
\end{table*}

\begin{table*}[htp]
    \centering
    \caption{Error rates for $\overline{\exp(-\ii\theta Z^{\otimes w}/2)}$ in the $[[5, 1, 3]]$ or the $[[7, 1, 3]]$ code with $q = 0.001$}\label{tab:perfect_multiq_comparison_1}
    \begin{tabular}{c@{\hspace{15pt}}c@{\hspace{15pt}}c@{\hspace{15pt}}c@{\hspace{15pt}}c}
        \hline\hline
        Weight & $p_U\,\,(\times 10^{-3})$ & $p_S$ (\%) & $p_L/p_S\,\,(\times 10^{-3})$ & $p_U / (p_L/p_S)$ \\
        \hline
        $w = 2$ & 3 &  99.01 & 1.14 & 2.62\\
        $w = 3$ & 5 &  98.41 & 1.15 & 4.34\\
        $w = 4$ & 7 &  97.81 & 1.16 & 6.04\\
        $w = 5$ & 9 &  97.21 & 1.17 & 7.72\\
        \hline
    \end{tabular}
\end{table*}

\begin{table*}[htp]
    \centering
    \caption{Error rates for $\overline{\exp(-\ii\theta Z^{\otimes w}/2)}$ in the $[[5, 1, 3]]$ and $[[7, 1, 3]]$ codes with $q = 0.0001$}\label{tab:perfect_multiq_comparison_2}
    \begin{tabular}{c@{\hspace{15pt}}c@{\hspace{15pt}}c@{\hspace{15pt}}c@{\hspace{15pt}}c}
        \hline\hline
        Weight & $p_U\,\,(\times 10^{-3})$ & $p_S$ (\%) & $p_L/p_S\,\,(\times 10^{-3})$ & $p_U / (p_L/p_S)$ \\
        \hline
        $w = 2$ & 2.1 & 99.01 & 0.236 & 8.91\\
        $w = 3$ & 4.1 & 98.41 & 0.237 & 17.29\\
        $w = 4$ & 6.1 & 97.81 & 0.239 & 25.57\\
        $w = 5$ & 8.1 & 97.21 & 0.240 & 33.75\\
        \hline
    \end{tabular}
\end{table*}

\begin{table*}[htp]
    \centering
    \caption{Error rates in the $[[15, 7, 3]]$ code with $q = 0.001$}\label{tab:hamming_multiq_comparison_1}
    \begin{tabular}{c@{\hspace{15pt}}c@{\hspace{15pt}}c@{\hspace{15pt}}c@{\hspace{15pt}}c@{\hspace{15pt}}c}
        \hline\hline
        Logical Pauli & Weight & $p_U\,\,(\times 10^{-3})$ & $p_S$ (\%) & $p_L/p_S\,\,(\times 10^{-3})$ & $p_U / (p_L/p_S)$ \\
        \hline
        $\overline{Z}_{4,5,6,7}$ &$w = 4, t = 4$ & 7 & 99.41 & 1.14 & 6.14\\
        $\overline{Z}_{1, 2}$ &$w = 2, t = 6$ & 3 & 99.01 & 1.14 & 2.62\\
        $\overline{Z}_{5, 6, 7}$ &$w = 3, t = 7$ & 5 & 98.81 & 1.15 & 4.36\\
        $\overline{Z}_{2, 3, 4, 5}$ &$w = 4, t = 8$ & 7 & 98.61 & 1.15 & 6.09\\
        $\overline{Z}_{1, 2, 3, 4, 5}$ &$w = 5, t = 9$ & 9 & 98.41 & 1.15 & 7.82\\
        $\overline{Z}_{1, 2, 3, 5, 6, 7}$ &$w = 6, t = 10$ & 11 & 98.21 & 1.15 & 9.53\\
        $\overline{Z}_{1, 2, 3}$ &$w = 3, t = 11$ & 5 & 98.01 & 1.16 & 4.32\\
        $\overline{Z}_{1, 2, 3, 4}$ &$w = 4, t = 12$ & 7 & 97.81 & 1.16 & 6.04\\
        \hline
    \end{tabular}
\end{table*}

\begin{table*}[htp]
    \centering
    \caption{Error rates in the $[[15, 7, 3]]$ code with $q = 0.0001$}\label{tab:hamming_multiq_comparison_2}
    \begin{tabular}{c@{\hspace{15pt}}c@{\hspace{15pt}}c@{\hspace{15pt}}c@{\hspace{15pt}}c@{\hspace{15pt}}c}
        \hline\hline
        Logical Pauli & Weight & $p_U\,\,(\times 10^{-3})$ & $p_S$ (\%) & $p_L/p_S\,\,(\times 10^{-3})$ & $p_U / (p_L/p_S)$ \\
        \hline
        $\overline{Z}_{4,5,6,7}$ &$w = 4, t = 4$ & 6.1 & 99.41 & 0.23 & 25.99\\
        $\overline{Z}_{1, 2}$ &$w = 2, t = 6$ & 2.1 & 99.01 & 0.24 & 8.91\\
        $\overline{Z}_{5, 6, 7}$ &$w = 3, t = 7$ & 4.1 & 98.81 & 0.24 & 17.36\\
        $\overline{Z}_{2, 3, 4, 5}$ &$w = 4, t = 8$ & 6.1 & 98.61 & 0.24 & 25.78\\
        $\overline{Z}_{1, 2, 3, 4, 5}$ &$w = 5, t = 9$ & 8.1 & 98.41 & 0.24 & 34.16\\
        $\overline{Z}_{1, 2, 3, 5, 6, 7}$ &$w = 6, t = 10$ & 10.1 & 98.21 & 0.24 & 42.51\\
        $\overline{Z}_{1, 2, 3}$ &$w = 3, t = 11$ & 4.1 & 98.01 & 0.24 & 17.22\\
        $\overline{Z}_{1, 2, 3, 4}$ &$w = 4, t = 12$ & 6.1 & 97.81 & 0.24 & 25.57\\
        \hline
    \end{tabular}
\end{table*}

% \section{Quantum Code Definition}\label{app:stab_codes}
% \input{appendix/2-codes}

% \section{Type II Circuit}\label{app:type_ii}
% \input{appendix/2-type-ii}

\end{document}